\documentclass[11pt,reqno]{amsart}

\usepackage[margin=1.15in]{geometry}
\usepackage{amsmath,amssymb,mathtools}
\usepackage{enumitem}
\usepackage{hyperref}
\usepackage[nameinlink,capitalize]{cleveref}

\hypersetup{
  colorlinks=true,
  linkcolor=black,
  citecolor=blue,
  urlcolor=blue
}

\newcommand{\ii}{\mathrm{i}}
\newcommand{\dd}{\mathrm{d}}

\newcommand{\cR}{\mathcal{R}}

\newcommand{\floor}[1]{\left\lfloor #1\right\rfloor}

\numberwithin{equation}{section}

\newtheorem{theorem}{Theorem}[section]
\newtheorem{proposition}[theorem]{Proposition}

\newtheorem{remark}[theorem]{Remark}

\title[Morse momentum wavefunctions and rational functions]
{The Morse potential as a physical realization of finite biorthogonal rational functions}

\author{Luc Vinet}
\address{Centre de recherches math\'ematiques, Universit\'e de Montr\'eal, P.O. Box 6128, Centre-ville Station, Montr\'eal, Qu\'ebec, H3C 3J7, Canada}
\email{luc.vinet@umontreal.ca}

\author{Alexei Zhedanov}
\address{Leonhard Euler International Mathematical Institute, Saint Petersburg, Russian Federation}

\begin{document}

\begin{abstract}
We revisit the bound states of the Morse potential in the momentum representation.
After the ground-state factor is extracted, the remaining factors are finite rational
functions of the momentum variable.  These functions are eigenfunctions of a
second-order difference operator and are identified with the symmetric specialization
of a finite family of biorthogonal rational functions introduced by Koepf and
Masjed-Jamei.  They also satisfy a generalized eigenvalue problem in the degree
variable, thereby placing the Morse momentum wavefunctions within the framework
of rational bispectrality and $R_{II}$-type systems.  Finally, after extraction of
their poles, the same wavefunctions are expressed in terms of Meixner--Pollaczek
polynomials with degree-dependent parameters.  This gives a simple description of
their zeros.  The Morse potential thus provides a concrete quantum-mechanical
realization of finite biorthogonal rational functions.
\end{abstract}

\maketitle

\section{Introduction}

Orthogonal polynomials arise naturally throughout exactly solvable quantum
mechanics.  The Hermite, Laguerre and Jacobi polynomials appear respectively in
the harmonic oscillator, Coulomb and P\"oschl--Teller systems; more generally,
hypergeometric orthogonal polynomials are ubiquitous in one-dimensional solvable
models and in superintegrable systems \cite{Chihara1978,Ismail2005,KoekoekLeskySwarttouw2010}.
In the discrete setting, finite Jacobi matrices, spin chains, discrete quantum
oscillators and quantum walks provide equally important realizations of
orthogonal polynomials.  In particular, Krawtchouk, Hahn, dual Hahn, Racah and
$q$-Racah polynomials occur in finite spin chains with perfect state transfer,
in discrete oscillator models and in the algebraic description of finite quantum
systems \cite{Godsil2012,VinetZhedanov2012a,VinetZhedanov2012b,GenestPostVinet2017,GenestVinetZhedanov2014}.

In contrast, explicit physical realizations of orthogonal or biorthogonal
rational functions are much less common.  Such functions arise naturally in
approximation theory, continued fractions, rational interpolation and generalized
eigenvalue problems (GEVP)\cite{IsmailMasson1994,Zhedanov1999, Rosengren, Rosengren2}.  They have also appeared
in the theory of Wilson functions and in recent algebraic treatments of
bispectral rational functions \cite{Groenevelt2003,TsujimotoVinetZhedanov2020,BussiereGaboriaudVinetZhedanov2022}.
They have also appeared as solutions of the GEVPs \cite{tsujimoto2023rational, tsujimoto2026ruijsenaars, tsujimoto2026rational} associated to the rational Heun operators identified as two-particle degenerations introduced by Takemura \cite{takemura2017degenerations} of the Ruijsenaars - Van Diejen Hamiltonians \cite{ruijsenaars2004integrable, van1994integrability}. Nevertheless, familiar exactly solvable quantum systems whose eigenfunctions are
naturally expressed in terms of finite rational functions remain rare.

The purpose of the present paper is to show that the bound-state wavefunctions of
the Morse potential in the momentum representation provide such an example.  In
coordinate space, the Morse potential is one of the classical exactly solvable
one-dimensional potentials and its bound states are expressed through confluent
hypergeometric functions.  In the momentum representation, however, the
Schr\"odinger equation becomes a second-order difference equation.  We shall see
that, after the ground-state momentum wavefunction is factored out, the remaining
factors are rational functions of the momentum variable.

The momentum-space Morse oscillator was studied in \cite{DahlSpringborg1988} and
an explicit hypergeometric expression for the bound states was obtained in
\cite{SunDong2012}.  Our emphasis here is different.  We show that the rational
factors in the momentum wavefunctions coincide with the symmetric specialization
of a finite family of biorthogonal rational functions introduced by Koepf and
Masjed-Jamei \cite{KoepfMasjedJamei2007}.  This identifies the Morse potential as
a physical realization of a finite rational system of $R_{II}$ type.

A second point is the bispectral nature of these functions.  They satisfy a
difference equation in the momentum variable and, dually, a generalized eigenvalue
problem in the degree variable.  This places the Morse momentum wavefunctions in
the setting of rational bispectrality, where the role played by ordinary
three-term recurrence relations for orthogonal polynomials is taken over by
tridiagonal matrix pencils.  In this sense, the role played here by finite rational
functions is analogous to the role played by Hermite, Laguerre, Jacobi,
Krawtchouk, Hahn and Racah polynomials in continuous and discrete exactly
solvable models.

Finally, by applying a classical transformation of the Gauss hypergeometric
function, we express the pole-free part of the Morse momentum wavefunctions in
terms of Meixner--Pollaczek polynomials.  This gives an immediate description of
the zeros: for the $n$th bound state, all $n$ zeros are simple and lie on the
horizontal line $\operatorname{Im} p=-1/2$ in the complex momentum plane.

The paper is organized as follows.  In \cref{sec:morse} we recall the Morse
potential and derive the finite-difference equation in momentum space.  In
\cref{sec:rational} we factor out the ground state and construct the rational
eigenfunctions.  In \cref{sec:kmj} we identify these functions with the
Koepf--Masjed-Jamei finite rational functions and derive their orthogonality.  In
\cref{sec:gevp} we formulate the corresponding generalized eigenvalue problem and
explain the resulting bispectrality.  In \cref{sec:mp} we relate the same
wavefunctions to Meixner--Pollaczek polynomials and determine their zeros.

\section{The Morse potential in momentum space}\label{sec:morse}

We consider the Morse potential
\begin{equation}
  U(x)=g^2\left(e^{-2x}-2e^{-x}\right),
  \qquad g>0.
  \label{eq:morse_potential}
\end{equation}
It tends to zero as $x\to\infty$ and has its minimum $U_{\min}=-g^2$ at $x=0$.
The Schr\"odinger equation for bound states is
\begin{equation}
  H\psi_n=E_n\psi_n,
  \label{eq:schrodinger_abstract}
\end{equation}
where
\begin{equation}
  H=p^2+U(x),
  \qquad p=-\ii\partial_x .
  \label{eq:hamiltonian}
\end{equation}
In coordinate representation this gives
\begin{equation}
  -\psi_n''(x)+U(x)\psi_n(x)=E_n\psi_n(x).
  \label{eq:coordinate_schrodinger}
\end{equation}
The bound-state energies are
\begin{equation}
  E_n=-\left(g-n-\frac12\right)^2,
  \qquad n=0,1,\dots,N,
  \qquad N=\floor{g-\frac12}.
  \label{eq:energies}
\end{equation}
The corresponding coordinate wavefunctions can be written as
\begin{equation}
  \psi_n(x)=G_n e^{-y/2} y^{g-n-1/2}
  {}_1F_1\!\left(
  \begin{matrix}-n\\ 2g-2n\end{matrix};y
  \right),
  \qquad y=2g e^{-x},
  \label{eq:coordinate_wavefunctions}
\end{equation}
where $G_n$ is a normalization constant.

The momentum representation is defined by the Fourier transform
\begin{equation}
  \Psi(p)=\frac{1}{\sqrt{2\pi}}
  \int_{-\infty}^{\infty} e^{-\ii px}\psi(x)\,\dd x.
  \label{eq:fourier}
\end{equation}
In momentum space the coordinate operator becomes
\begin{equation}
  x=\ii\partial_p.
\end{equation}
Consequently the exponentials in the Morse potential act as shift operators,
\begin{equation}
  e^{-x}=T,
  \qquad e^{-2x}=T^2,
  \label{eq:shift_exp}
\end{equation}
where
\begin{equation}
  T^jF(p)=F(p-j\ii),
  \qquad j=0,\pm1,\pm2,\dots .
  \label{eq:shift_def}
\end{equation}
Thus the Hamiltonian becomes
\begin{equation}
  H=p^2+g^2(T^2-2T),
  \label{eq:momentum_hamiltonian}
\end{equation}
and the bound-state equation reads
\begin{equation}
  p^2\Psi_n(p)+g^2\bigl(\Psi_n(p-2\ii)-2\Psi_n(p-\ii)\bigr)
  =E_n\Psi_n(p).
  \label{eq:momentum_difference_eq}
\end{equation}
The explicit expression obtained in \cite{SunDong2012} is
\begin{equation}
  \Psi_n(p)=K_n\Gamma\!\left(g-n-\frac12+\ii p\right)g^{-\ii p}
  {}_2F_1\!\left(
  \begin{matrix}
  -n,\; g-n-\frac12+\ii p\\
  2g-2n
  \end{matrix};2
  \right),
  \label{eq:sun_dong}
\end{equation}
where $K_n$ is independent of $p$.

\section{Rational eigenfunctions}\label{sec:rational}

The ground-state momentum wavefunction is, up to normalization,
\begin{equation}
  \Psi_0(p)=g^{-\ii p}\Gamma\!\left(g-\frac12+\ii p\right).
  \label{eq:ground_state}
\end{equation}
It is a solution of \eqref{eq:momentum_difference_eq} with energy
\begin{equation}
  E_0=-\left(g-\frac12\right)^2.
\end{equation}
We write the excited states in the form
\begin{equation}
  \Psi_n(p)=\Psi_0(p)R_n(p).
  \label{eq:factorization}
\end{equation}
Substitution in \eqref{eq:momentum_difference_eq} gives the eigenvalue equation
\begin{equation}
  W R_n(p)=E_n R_n(p),
  \label{eq:W_eigenvalue}
\end{equation}
where the difference operator $W$ is
\begin{equation}
  W=p^2I+\bigl((g+\ii p)^2-\frac14\bigr)T^2
  -2g\left(g-\frac12+\ii p\right)T.
  \label{eq:W_operator}
\end{equation}
In particular,
\begin{equation}
  W\{1\}=E_0.
\end{equation}

Let
\begin{equation}
  \chi_k(p)=\frac{1}{\left(-\ii p-g+\frac32\right)_k},
  \qquad k=0,1,2,\dots,
  \label{eq:chi_basis}
\end{equation}
where
\begin{equation}
  (a)_k=a(a+1)\cdots(a+k-1),
  \qquad (a)_0=1.
\end{equation}
For each $n$, the functions $\chi_0,\chi_1,\dots,\chi_n$ span the finite-dimensional space
\begin{equation}
  \cR_n=\operatorname{span}\{\chi_0(p),\chi_1(p),\dots,\chi_n(p)\}.
\end{equation}
Equivalently, elements of $\cR_n$ are rational functions of the form
\begin{equation}
  \frac{Q_n(p)}{
  \left(-\ii p-g+\frac32\right)
  \left(-\ii p-g+\frac52\right)
  \cdots
  \left(-\ii p-g+\frac{2n+1}{2}\right)},
  \label{eq:rational_flag}
\end{equation}
where $Q_n(p)$ is a polynomial of degree at most $n$.

\begin{proposition}\label{prop:bidiagonal}
The operator $W$ acts on the basis \eqref{eq:chi_basis} as
\begin{equation}
  W\chi_k(p)=
  -\left(g-k-\frac12\right)^2\chi_k(p)+2k\chi_{k-1}(p),
  \qquad k\ge 0,
  \label{eq:bidiagonal_action}
\end{equation}
with the convention $\chi_{-1}=0$.
\end{proposition}

\begin{proof}
This follows by direct substitution.  Under the shift $T$, the variable $\ii p$ is replaced by
$\ii p+1$, and hence
\begin{equation}
  T\chi_k(p)=\frac{1}{\left(-\ii p-g+\frac12\right)_k},
  \qquad
  T^2\chi_k(p)=\frac{1}{\left(-\ii p-g-\frac12\right)_k}.
\end{equation}
Substituting these expressions in \eqref{eq:W_operator} and reducing the resulting rational
function to the basis $\{\chi_k,\chi_{k-1}\}$ gives \eqref{eq:bidiagonal_action}.
\end{proof}

It follows that, for each $n=0,1,\dots,N$, there is an eigenfunction of $W$ in $\cR_n$ with
eigenvalue $E_n$.  Write
\begin{equation}
  R_n(p)=\sum_{k=0}^{n} A_{n,k}\chi_k(p).
  \label{eq:R_expansion}
\end{equation}
Using \eqref{eq:bidiagonal_action} in \eqref{eq:W_eigenvalue} gives the recurrence
\begin{equation}
  A_{n,k+1}
  =\frac{(k-n)(k+n+1-2g)}{2(k+1)}A_{n,k},
  \qquad k=0,1,\dots,n-1.
  \label{eq:A_recurrence}
\end{equation}
Choosing $A_{n,0}=1$, one obtains
\begin{equation}
  A_{n,k}=\frac{(-n)_k(n+1-2g)_k}{k!}\left(\frac12\right)^k.
  \label{eq:A_coefficients}
\end{equation}
Hence
\begin{equation}
  R_n(p)=
  {}_2F_1\!\left(
  \begin{matrix}
  -n,\; n+1-2g\\
  -\ii p-g+\frac32
  \end{matrix};\frac12
  \right).
  \label{eq:R_hypergeometric}
\end{equation}
Thus
\begin{equation}
  \Psi_n(p)=\kappa_n\Psi_0(p)R_n(p)
  \label{eq:Psi_R_solution}
\end{equation}
provides the bound-state solutions of \eqref{eq:momentum_difference_eq}, with energies
\eqref{eq:energies}.

\section{Identification with finite biorthogonal rational functions}\label{sec:kmj}

The rational functions \eqref{eq:R_hypergeometric} are naturally related to the finite family of
rational functions introduced by Koepf and Masjed-Jamei \cite{KoepfMasjedJamei2007}.  In the
notation used here, that family is
\begin{equation}
  B_n(x;a,b)=
  {}_2F_1\!\left(
  \begin{matrix}
  -n,\; n-a-b\\
  1-a-x
  \end{matrix};\frac12
  \right).
  \label{eq:KMJ_functions}
\end{equation}
The functions $B_n$ satisfy a finite biorthogonality relation of the form
\begin{equation}
  \int_{-\infty}^{\infty}
  \Gamma(a+\ii x)\Gamma(b-\ii x)
  B_n(\ii x;a,b)B_m(-\ii x;b,a)\,\dd x=0,
  \qquad n\ne m,
  \label{eq:KMJ_biorthogonality}
\end{equation}
for the admissible finite range of indices.  In the symmetric specialization $a=b$, this
biorthogonality becomes an ordinary orthogonality relation.

Comparing \eqref{eq:R_hypergeometric} and \eqref{eq:KMJ_functions}, one obtains
\begin{equation}
  x=\ii p,
  \qquad a=b=g-\frac12.
  \label{eq:parameter_identification}
\end{equation}
Equivalently,
\begin{equation}
  R_n(p)=B_n\!\left(\ii p;g-\frac12,g-\frac12\right).
  \label{eq:R_B_identification}
\end{equation}

\begin{theorem}\label{thm:main_identification}
The rational factors of the momentum-space bound states of the Morse potential coincide with the
symmetric Koepf--Masjed-Jamei finite rational functions:
\begin{equation}
  R_n(p)=B_n\!\left(\ii p;g-\frac12,g-\frac12\right),
  \qquad n=0,1,\dots,N.
\end{equation}
Consequently, the Morse momentum wavefunctions provide a physical realization of a finite family
of orthogonal rational functions belonging to the $R_{II}$ framework.
\end{theorem}

The physical orthogonality of the Morse bound states gives the corresponding orthogonality
relation for these rational functions.  With the normalization chosen so that the wavefunctions are
orthonormal,
\begin{equation}
  \int_{-\infty}^{\infty}
  \overline{\Psi_m(p)}\Psi_n(p)\,\dd p=\delta_{mn}.
\end{equation}
Using \eqref{eq:factorization}, this becomes
\begin{equation}
  \int_{-\infty}^{\infty}
  w(p)R_n(p)\overline{R_m(p)}\,\dd p=h_n\delta_{mn},
  \label{eq:R_orthogonality_physical}
\end{equation}
where
\begin{equation}
  w(p)=|\Psi_0(p)|^2
  =\left|\Gamma\!\left(g-\frac12+\ii p\right)\right|^2.
  \label{eq:weight}
\end{equation}
Here the factor $g^{-\ii p}$ has modulus one for real $p$ and $g>0$.  Since
$\overline{R_m(p)}=R_m(-p)$ for real $p$, \eqref{eq:R_orthogonality_physical} is the symmetric
specialization of the Koepf--Masjed-Jamei biorthogonality.

\begin{remark}
The terminology $R_{II}$ is often attached to rational functions whose recurrence relations are
encoded by tridiagonal generalized eigenvalue problems.  In its standard form, an $R_{II}$
recurrence contains a quadratic factor multiplying the preceding element.  Equivalently, it can be
written as a tridiagonal matrix pencil.  The Koepf--Masjed-Jamei functions form a finite
hypergeometric realization of this rational framework; the Morse model realizes the symmetric
orthogonal specialization of that family.
\end{remark}

\section{Generalized eigenvalue problem and bispectrality}\label{sec:gevp}

The rational functions $R_n(p)$ satisfy two eigenvalue problems.  The first is the difference
equation \eqref{eq:W_eigenvalue}, where the operator $W$ acts on the momentum variable $p$ and
the eigenvalue is the energy $E_n$.

The second is a generalized eigenvalue problem in the degree variable.  This is the rational
analogue of the three-term recurrence relation for orthogonal polynomials.  In the theory of
biorthogonal rational functions, such recurrences are naturally written as matrix pencils
\begin{equation}
  J_1 R=zJ_2 R,
  \label{eq:abstract_GEVP}
\end{equation}
where $J_1$ and $J_2$ are tridiagonal matrices \cite{Zhedanov1999}.  For the Morse rational
functions, the spectral parameter is $z=\ii p$ and the dual problem is
\begin{equation}
  J_1R(p)=\ii p\,J_2R(p),
  \label{eq:morse_GEVP}
\end{equation}
where
\begin{equation}
  R(p)=\bigl(R_0(p),R_1(p),\dots,R_N(p)\bigr)^T.
\end{equation}
In components, for the interior values of the finite system,
\begin{align}
  &\eta_{1,n}R_{n+1}(p)+\eta_{2,n}R_{n-1}(p)+\eta_{0,n}R_n(p)
  \notag\\
  &\qquad
  =-\ii p\left(
  \xi_{1,n}R_{n+1}(p)+\xi_{2,n}R_{n-1}(p)+\xi_{0,n}R_n(p)
  \right).
  \label{eq:component_GEVP}
\end{align}
The coefficients are obtained from contiguous relations for the Gauss hypergeometric function:
\begin{align}
  \eta_{2,n}&=n-g-\frac12,
  &
  \eta_{1,n}&=
  \frac{(g-n)(2g-n-1)(g-n-\frac32)}{n(g-n-1)},
  &
  \eta_{0,n}&=
  \frac{g(g-n-\frac12)}{n(g-n-1)},
  \label{eq:eta_coefficients}
\end{align}
while
\begin{align}
  \xi_{2,n}&=1,
  &
  \xi_{1,n}&=
  \frac{(g-n)(2g-n-1)}{n(g-n-1)},
  &
  \xi_{0,n}&=-
  \frac{g(2g-2n-1)}{n(g-n-1)}.
  \label{eq:xi_coefficients}
\end{align}
Boundary terms are obtained by truncating the finite system at $n=0$ and $n=N$.

Equations \eqref{eq:W_eigenvalue} and \eqref{eq:morse_GEVP} exhibit the bispectral character of
the Morse rational functions.  The difference operator $W$ acts in the momentum variable and has
spectrum $E_n$, whereas the generalized eigenvalue problem acts in the degree variable and has
spectral parameter $\ii p$.  This is the rational counterpart of the bispectrality of classical
orthogonal polynomials, with tridiagonal matrix pencils replacing ordinary three-term recurrence
relations.

\section{Relation with Meixner--Pollaczek polynomials and zeros}\label{sec:mp}

There is another useful representation of the momentum wavefunctions.  We use the transformation
formula
\begin{equation}
  {}_2F_1\!\left(
  \begin{matrix}
  -n,\;\beta\\
  \gamma
  \end{matrix};z
  \right)
  =\frac{(\beta)_n}{(\gamma)_n}(-z)^n
  {}_2F_1\!\left(
  \begin{matrix}
  -n,\;1-\gamma-n\\
  1-\beta-n
  \end{matrix};z^{-1}
  \right).
  \label{eq:hypergeometric_transform}
\end{equation}
Applying this transformation to \eqref{eq:R_hypergeometric} gives
\begin{equation}
  R_n(p)=
  \frac{(-1/2)^n(n+1-2g)_n}{(3/2-g-\ii p)_n}
  V_n(p),
  \label{eq:R_V_factorization}
\end{equation}
where
\begin{equation}
  V_n(p)=
  {}_2F_1\!\left(
  \begin{matrix}
  -n,\; g-n+\ii p-\frac12\\
  2g-2n
  \end{matrix};2
  \right).
  \label{eq:V_def}
\end{equation}
The function $V_n(p)$ is a polynomial in $p$ of degree $n$.  In this form all poles of
$R_n(p)$ are contained in the factor $(3/2-g-\ii p)_n^{-1}$, and \eqref{eq:V_def} agrees with the
hypergeometric expression for the momentum wavefunction displayed in \eqref{eq:sun_dong}.

The Meixner--Pollaczek polynomials are defined, up to a normalization factor independent of the
argument, by
\begin{equation}
  P_n(t;\mu,\phi)=
  {}_2F_1\!\left(
  \begin{matrix}
  -n,\;\mu+\ii t\\
  2\mu
  \end{matrix};1-e^{-2\ii\phi}
  \right).
  \label{eq:MP_def}
\end{equation}
Comparison of \eqref{eq:V_def} with \eqref{eq:MP_def} gives
\begin{equation}
  \phi=\frac{\pi}{2},
  \qquad \mu=g-n,
  \qquad t=p+\frac{\ii}{2}.
  \label{eq:MP_parameters}
\end{equation}
Therefore
\begin{equation}
  V_n(p)=P_n\!\left(p+\frac{\ii}{2};g-n,\frac{\pi}{2}\right),
  \label{eq:V_MP}
\end{equation}
up to the normalization convention used for the Meixner--Pollaczek polynomials.

For $\mu>0$, the Meixner--Pollaczek polynomials are positive definite and have $n$ simple real
zeros.  In the present case $\mu=g-n$, and for bound states
$n=0,1,\dots,\floor{g-1/2}$ one has $g-n>0$.
Consequently, the zeros of $V_n(p)$ are obtained from the real zeros of a positive definite
Meixner--Pollaczek polynomial by the shift $t=p+\ii/2$.

\begin{theorem}\label{thm:zeros}
For each bound state $n=0,1,\dots,N$, the momentum-space wavefunction $\Psi_n(p)$ has exactly
$n$ simple zeros.  They lie on the horizontal line
\begin{equation}
  \operatorname{Im}p=-\frac12
\end{equation}
in the complex $p$-plane.  In particular, $\Psi_n(p)$ has no zeros on the real momentum axis.
\end{theorem}

\begin{proof}
The ground-state factor \eqref{eq:ground_state} has no zeros, and the factor
$(3/2-g-\ii p)_n^{-1}$ in \eqref{eq:R_V_factorization} has no zeros.  Hence the zeros of
$\Psi_n(p)$ are precisely the zeros of $V_n(p)$.  By \eqref{eq:V_MP}, these are obtained from the
real zeros of the Meixner--Pollaczek polynomial $P_n(t;g-n,\pi/2)$ by the shift $t=p+\ii/2$.
Thus real zeros in $t$ correspond to zeros of $\Psi_n(p)$ lying on
$\operatorname{Im}p=-1/2$.  Simplicity follows from positive definiteness of the
Meixner--Pollaczek orthogonality for $g-n>0$.
\end{proof}

\begin{remark}
Although for each fixed $n$ the polynomial $V_n(p)$ is a Meixner--Pollaczek polynomial with shifted
argument, the sequence $V_0,V_1,\dots,V_N$ is not a single ordinary Meixner--Pollaczek family.
Indeed, the parameter $\mu=g-n$ depends on the degree.  Consequently these polynomials do not
satisfy the standard three-term recurrence relation of an ordinary orthogonal polynomial system.
The recurrence structure of the original rational factors $R_n(p)$ is instead expressed through
the generalized eigenvalue problem \eqref{eq:morse_GEVP}.
\end{remark}

\section{Conclusion}

We have shown that the momentum-space bound states of the Morse potential provide a natural
realization of a finite family of orthogonal rational functions.  After extraction of the
ground-state factor, the corresponding rational functions coincide with the symmetric
Koepf--Masjed-Jamei family.  They satisfy both a difference equation in the momentum variable and
a generalized eigenvalue problem in the degree index, thereby forming a bispectral system of
$R_{II}$ type.

The additional identification with Meixner--Pollaczek polynomials, after the poles have been
extracted, yields a simple description of the zeros of the momentum wavefunctions: all zeros are
simple and lie on a horizontal line in the complex momentum plane.

From a broader perspective, the Morse potential appears as one of the few familiar exactly
solvable quantum systems whose eigenfunctions are naturally expressed in terms of finite rational
functions rather than polynomial families.  It therefore provides a physical counterpart to the
role played by Hermite, Laguerre, Jacobi, Krawtchouk, Hahn and Racah polynomials in continuous
and discrete exactly solvable models.  This example also illustrates the relevance of generalized
eigenvalue problems in quantum mechanics and suggests that other solvable systems may give rise
to rational bispectral functions.

\section*{Acknowledgments}
The authors are grateful to Nicolas Crampé for illuminating discussions. LV is funded in part through a discovery grant of the Natural Sciences and Engineering Research Council (NSERC) of Canada. 
 AZ is supported by the Ministry of Science and Higher Education of the Russian Federation (agreement no. 075–15–2025–343).

\section*{Conflict of interest}
The authors state that there is no conflict of interest.

\section*{Data availability}
This manuscript has no associated data.

\end{document}